\documentclass{amsart}
\usepackage[mathscr]{eucal}
\usepackage{amsmath,amsfonts,amssymb}
\usepackage{graphicx}
\newtheorem{theorem}{Theorem}[section]

\newtheorem{proposition}[theorem]{Proposition}

\theoremstyle{definition}
\newtheorem{definition}[theorem]{Definition}
\newtheorem{example}[theorem]{Example}
\newtheorem{remark}[theorem]{Remark}
\makeatletter
\newcommand{\card}{\mathop{\operator@font card}}
\makeatother
\renewcommand{\theenumi}{\roman{enumi}}

\makeatletter
\renewcommand\p@enumii{(\theenumi)}
\def\initmark{{\it Iranian Journal of Fuzzy Systems\/} {\bf Vol. 9, No. 2,} (2012) pp. 113-125}

\def\ps@headings{\ps@empty
  \def\@evenhead{\normalfont\scriptsize
      \rlap{\thepage}\hfil \leftmark{}{}\hfil}%
  \def\@oddhead{\normalfont\scriptsize \hfil
      \rightmark{}{}\hfil \llap{\thepage}}%
  \let\@mkboth\markboth
}

\def\ps@firstpage{\ps@empty
  \def\@oddhead{\normalfont\scriptsize
      \initmark{}{}\hfil \llap{\thepage}}%
  \let\@evenhead\@oddhead 
}

\def\@maketitle{%
  \normalfont\normalsize
  \let\@makefnmark\relax  \let\@thefnmark\relax
  \ifx\@empty\@subjclass\else \@footnotetext{\@setsubjclass}\fi
  \ifx\@empty\@keywords\else \@footnotetext{\@setkeywords}\fi
  \ifx\@empty\thankses\else \@footnotetext{%
    \def\par{\let\par\@par}\@setthanks}\fi
  \@mkboth{\@nx\shortauthors}{\@nx\shorttitle}%
  \global\topskip42\p@\relax 
  \@settitle
  \ifx\@empty\authors \else \@setauthors \fi
  \ifx\@empty\@date\else \@setdate \fi 
  \ifx\@empty\@dedicatory
  \else
    \baselineskip18\p@
    \vtop{\centering{\footnotesize\itshape\@dedicatory\@@par}%
      \global\dimen@i\prevdepth}\prevdepth\dimen@i
  \fi
  \@setabstract
  \normalsize
  \if@titlepage
    \newpage
  \else
    \dimen@16\p@ \advance\dimen@-\baselineskip
    \vskip\dimen@\relax
  \fi
} 

\def\@setdate{%
  \begingroup
  \trivlist
  \centering\footnotesize \@topsep30\p@\relax
  \advance\@topsep by -\baselineskip
  \item\relax
\@date
  \endtrivlist
  \endgroup
}

\def\@setabstracta{%
  \ifvoid\abstractbox
  \else
    \skip@20\p@ \advance\skip@-\lastskip
    \advance\skip@-\baselineskip \vskip\skip@
    \box\abstractbox
    \prevdepth\z@ 
  \fi
}

\makeatother
\def\rightmark{On Generalized Fuzzy Multisets and their Use in Computation}
\def\leftmark{A. Syropoulos}
\begin{document}
\title{On Generalized Fuzzy Multisets and their Use in Computation}
\author{Apostolos Syropoulos}
\address{Greek Molecular Computing Group\\
        366, 28th October St.\\
        GR-671\ 00\ \ Xanthi\ GREECE}
        \email{asyropoulos@yahoo.com}

\thanks{\textbf{Dedicated to the fond memory of my beloved mother Vassiliki Syropoulos.}\newline
\indent{\footnotesize Received: June 2010; Revised: March 2011; Accepted: May 2011}
}
\maketitle
\begin{abstract}
An orthogonal approach to the fuzzification of both multisets and hybrid
sets is presented. In particular, we introduce $L$-multi-fuzzy and 
$L$-fuzzy hybrid sets, which are general enough and in spirit with the
basic concepts of fuzzy set theory. In addition, we study the properties of
these structures. Also, the usefulness of these structures is examined in
the framework of mechanical multiset processing. More specifically, we 
introduce a variant of fuzzy P~systems and, since simple
fuzzy membrane systems have been introduced elsewhere, we simply extend
previously stated results and ideas.
\end{abstract}
\section{Introduction}
Intuitively, a set is a collection of elements (e.g., numbers or symbols) that
is completely determined by them.\footnote{For the present discussion this 
vague definition is adequate, but it may lead to paradoxes like the ``set of 
all sets'' paradox, which is known in the literature as Russell's paradox. 
However, such paradoxes will not concern us here.}  The elements of a set are 
pairwise different. If we relax this restriction and allow repeated 
occurrences of any element, then we end up with a mathematical structure that 
is known as {\em multiset}\footnote{The term ``multiset'' has been coined by 
N.G.~de~Bruijn~\cite{knuth81}.} (see~\cite{blizard91} for a historical account
of the development of the multiset theory; also, see~\cite{syropoulos01} for a
recent account of the mathematical theory of multisets). Multisets are really 
useful structures and they have found numerous applications in mathematics and 
computer science. For example, the prime factorization of an integer $n>0$ is 
a multiset $\mathcal{N}$ whose elements are primes. Also, every monic 
polynomial $f(x)$ over the complex numbers corresponds in a natural way to the
multiset $\mathcal{F}$ of its roots. In addition, multisets have been used in 
concurrency theory~\cite{deNicola96}. A rather interesting recent development
in the theory of multisets is the discovery that the logic of multisets is the
$\{\otimes,\multimap,\oplus,\mathbf{1}\}$-fragment of intuitionistic linear 
logic (see~\cite{tzouvaras98,tzouvaras03} for details).

If we allow elements of a multiset to occur an integral number of times (and 
that includes a {\em negative} number of times), we end up with a structure that
has been dubbed {\em hybrid set}. These structures have been introduced by 
Loeb~\cite{loeb92}. Initially, one may wonder whether hybrid sets are of any 
use. However, Loeb has shown that they are indeed very useful structures 
(see~\cite{loeb92,loeb95}). For example, one can use a hybrid set to describe 
the roots and the poles of a rational function. In particular, if $f(x)$ is a 
monic rational function, then $f(x)$ can be written  in terms of its roots 
$a_1,a_2,\ldots,a_n$, and its poles $b_1,b_2,\ldots,b_m$, as follows:
\begin{displaymath}
f(x)=c\frac{(x-a_1)(x-a_2)\cdots(x-a_n)}{(x-b_1)(x-b_2)\cdots(x-b_m)}
\end{displaymath}
From this we can directly form a hybrid set, where elements that occur a 
positive number of times correspond to the roots of the function and
elements that occur a negative number of times corresponds to the poles of 
the function.

In a seminal paper, Yager~\cite{yager86} introduced fuzzy multisets, that is
fuzzy subsets where an element may occur more than one time 
(see~\cite{miyamoto04} for an up-to-date presentation of the theory of fuzzy 
multisets, which, however, does not differ significantly 
from~\cite{miyamoto01}). Yager defined fuzzy multisets as follows~\cite{yager86}:
\begin{definition}
Assume $X$ is a set of elements. Then a fuzzy bag\footnote{Multisets are also known
as ``bags,'' ``heaps,'' ``bunches,'' ``samples,'' ``occurrence sets,''
``weighted sets,'' and ``firesets''---finitely repeated element sets.} $A$ 
drawn from $X$ can be characterized by a function 
$\mathrm{Count}.\mathrm{Mem}_{A}$ such that
\begin{displaymath}
\mathrm{Count}.\mathrm{Mem}_{A}:X\rightarrow Q,
\end{displaymath} 
where $Q$ is the set of all crisp bags drawn from the unit interval.
\end{definition}
In ``modern parlance'' fuzzy multisets can be characterized by a high-order 
function. In particular, a fuzzy multiset $A$ can be characterized by a 
function
\begin{displaymath}
A:X\rightarrow\mathbb{N}^{\mathrm{I}},
\end{displaymath}  
where $\mathrm{I}=[0,1]$ and $\mathbb{N}$ is the set of natural numbers
including zero. It is not difficult to see that any fuzzy multiset $A$ is 
actually characterized by a function 
\begin{displaymath}
A:X\times\mathrm{I}\rightarrow\mathbb{N},
\end{displaymath}
which is obtained from the former function by {\em uncurrying} it. However,
it is more natural to demand that for each element $x$ there is only one 
membership degree and one multiplicity. In other words, a ``fuzzy multiset'' $A$ 
should be characterized by a function $X\rightarrow\mathrm{I}\times\mathbb{N}$. 
To distinguish these structures from fuzzy multisets, we will call them {\em multi-fuzzy}
sets~\cite{syropoulos06}. Given a multi-fuzzy set, $A$, the expression $A(x)=(i,n)$ denotes 
that there are $n$ copies of $x$ that belong to $A$ with degree that is equal to $i$. 

Apart from their applicability to mathematics, multisets are really
useful structures as interesting models of computations are built upon them.
For instance, the chemical abstract machine of Berry and Boudol~\cite{cham}
is an abstract machine that is well-suited to model
concurrent computation and manipulates {\em solutions}, which are finite
multisets of {\em molecules} where a molecule is simply a term of an algebra.

Membrane computing is a model of computation that is built around 
the notion of multiset rewriting rules. More specifically, membrane computing 
is a computational paradigm that  was inspired by the way cells live and 
function (see~\cite{paun02} for an overview of the field of membrane 
computing). Roughly speaking, a cell consists of a membrane that separates the
cell from its environment. In addition, this membrane consists 
of compartments surrounded by porous membranes, which, in turn, may contain other 
compartments, and so on. At any moment, matter flows from one compartment to 
any neighboring one. In addition, the cell interacts with its environment in 
various ways (e.g., by dumping matter to its environment). Obviously, at any 
moment a number of processes occur in parallel (e.g., matter moves into a 
compartment, while energy is consumed in another compartment, etc.).

A P~system is a conceptual computational device whose functionality is based on an 
abstraction of the cell. Thus, a P~system consists of porous membranes that are 
populated with multisets of objects, which are usually materialized as strings
of symbols. In addition, there are rules that are used to change the 
configuration of the system. A P~system behaves more or less like a parser, 
which is clearly hard-wired to a particular grammar. Thus, a P~system stops 
when no rule can be applied to the system. The result of the computation is 
always equal to the cardinality of the multiset that is contained in a 
designated compartment. Now, since rigid mathematical models employed in life 
sciences are not completely adequate for the interpretation of biological 
information, there have been various proposals to use fuzzy sets in the 
modeling of biological systems (e.g., see~\cite{akay97,baum98}). Thus, it is 
quite reasonable to attempt the use of the theory of fuzzy sets in P~systems. 
Indeed, such an attempt has been described in~\cite{syropoulos06} that
is generalized to a certain degree in this paper.

\paragraph{Structure of the paper} In what follows, I will define 
$L$-multi-fuzzy sets and $L$-fuzzy hybrid sets. Next, I will define the basic 
operations between such structures (e.g., union, sum, etc.). Also, I  will 
give the definition of certain  standard fuzzy-theoretic operators. By
replacing multisets with either $L$-multi-fuzzy sets or $L$-hybrid sets in the
definition of both P~systems and the chemical abstract machine, we end up with
fuzzy versions of these notational computing devices. We formally define
these devices and briefly investigate their properties. The paper ends 
with the customary concluding remarks.
\section{On $L$-Multi Fuzzy Sets and $L$-Fuzzy Hybrid Sets}
One may say that multisets form an abstraction of the {\em token-type}
distinction, which is the basis of the ``token-token identity 
theory''~\cite{searle04}, while (ordinary) sets are an abstraction of the 
denial of the token-type distinction. To make clear the essence of the 
token-type distinction, I will borrow an example from~\cite{searle04}. If 
one writes the word ``dog'' three times (i.e., ``dog dog dog''), then she 
has written three instances, or tokens, of the one type of word. This observation 
necessitates a distinction between types (i.e, abstract general entities) and tokens
(i.e., concrete particular objects and events). ``A token of a type is
a particular concrete exemplification of that abstract general 
type''.~\cite[page 59]{searle04}. 

There is no question that three instances of the word ``dog'' are tokens 
of the ``dog'' type. However, there are many instances where one cannot make
such a definitive statement regarding the type of some tokens. In particular,
there are many cases where some token $t$ is of type $T$ to a certain degree.
For example, consider the following glyphs:
\begin{center}
\includegraphics[scale=1]{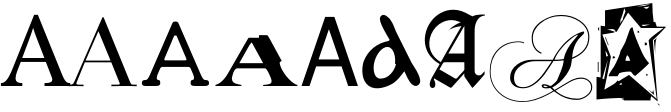}
\end{center}
Each of them depicts an ``A,'' however, each one is reminiscent of the ``standard A'' to 
some degree. For instance, the rightmost is less reminiscent of the ``standard A'', while
the leftmost glyph looks like an ordinary ``A.'' A text is clearly a multiset of letters,
but if we are free to use any of these ``As'' to typeset a document, then we need a
fuzzy multiset to {\em describe} the text. However, this is not a common practice and
so we need a more restricted structure that better models real life 
situations.\footnote{Although our example makes it clear that fuzzy multisets are useful,
still their usage example is not that realistic. Nevertheless, if we consider ``identical'' 
computational processes that may be similar to some {\em prototype} process to different degrees, 
then we have a situation where fuzzy multisets are useful (see~\cite{syropoulos2009} for more details regarding this idea).} As was noted above these structures will be called multi-fuzzy sets. 

Clearly, it is too restrictive to demand that tokens are of some type to a
degree, which is expressed by some number that belongs to the unit interval. 
More generally, we can assume that the likelihood degree is drawn from some
frame $L$. Indeed, it is possible to define many other partially ordered sets
that are frames. For example, Vickers~\cite{vickers90} shows that finite
observations on bit streams have the properties of a frame. Thus, the similarity
degree would express the idea that an element resembles, in some way, an element 
of such a frame. Typically, a frame is defined as follows~\cite{vickers90}:
\begin{definition}
A poset $A$ is a {\em frame} iff
\begin{enumerate}
\item every subset has a join
\item every finite subset has a meet
\item binary meets distribute over joins:
\begin{displaymath}
x\wedge\bigvee Y= \bigvee\Bigl\{ x\wedge y: y\in Y\Bigr\}.
\end{displaymath}
\end{enumerate}
\end{definition}
Note that a frame is clearly a distributive lattice.
So, $L$-multi-fuzzy sets are an extension of
multi-fuzzy sets just like $L$-fuzzy sets~\cite{goguen67} are an extension
of fuzzy sets. 

A.I. Kostrikin in his comments in the entry for the conecpt of duality in the
{\em Encyclopaedia of Mathematics}\footnote{See \texttt{http://eom.springer.de/}.} notes
that ``[d]uality is a very pervasive and important concept in (modern) mathematics.''
One could argue that hybrid sets, and, therefore, fuzzy hybrid sets, extend multisets and
their fuzzy counterparts to describe and/or to model dualities. Let us now proceed with the 
formal definition of $L$-fuzzy hybrid sets:
\begin{definition}\label{fuzzy:hybrid:set}
An $L$-fuzzy hybrid set $\mathscr{A}$ is a mathematical structure that is
characterized by a function $\mathscr{A}:X\rightarrow L\times\mathbb{Z}$, 
where $L$ is a frame, and it is associated with a $L$-fuzzy set 
$A:X\rightarrow L$. More specifically, the equality $\mathscr{A}(x)=(\ell,n)$ 
means that $\mathscr{A}$ contains exactly $n$ copies of $x$, where
$A(x)=\ell$.  
\end{definition}
If we substitute $\mathbb{Z}$ with $\mathbb{N}$ in the previous 
definition, then the resulting structures will be called {\em $L$-multi-fuzzy sets}.   

Assuming that $\mathscr{A}$ is an $L$-fuzzy hybrid set, then one can define 
the following two functions: the \textit{multiplicity} function 
$\mathscr{A}_{m}:X\rightarrow\mathbb{Z}$ and the \textit{membership} 
function $\mathscr{A}_{\mu}:X\rightarrow L$. Clearly, if 
$\mathscr{A}(x)=(\ell,n)$, then  $\mathscr{A}_{m}(x)=n$ 
and $\mathscr{A}_{\mu}(x)=\ell$.
Notice that it is equally easy to define the corresponding functions for an
$L$-multi-fuzzy set.

I believe this is a good point to briefly express my prejudices and my 
intentions regarding the present work. Clearly, it is not my intention to 
develop an axiomatic set theory of $L$-fuzzy hybrid sets and $L$-multi-fuzzy 
sets, in the sense of the Zermelo-Fr{\ae}nkel set theory, but rather a 
``na{\"\i}ve'' set theory in the sense that I will not present a precise 
axiomatization. Therefore, I plan to introduce only the basic set-theoretic 
operations and the basic properties of these sets. To begin with,
let me now define the cardinality of an $L$-fuzzy hybrid set: 
\begin{definition}
Assume that $\mathscr{A}$ is an $L$-fuzzy hybrid set that draws elements 
from a universe $X$. Then its cardinality is defined as follows:
\begin{displaymath}
\card\mathscr{A} = \sum_{x\in X}\mathscr{A}_{\mu}(x)\otimes\mathscr{A}_{m}(x),
\end{displaymath}
where $\otimes:L\times\mathbb{Z}\rightarrow\mathbb{R}$ is a binary multiplication operator 
that is used to compute the product of $\ell\in L$ times $n\in\mathbb{Z}$.
\end{definition}
\begin{example}
If $L=\mathrm{I}\times\mathrm{I}$ (i.e, when extenting ``intuitionistic'' fuzzy sets, 
see~\cite{syropoulos2010}), 
then $(i,j)\otimes n=in-jn$.
\end{example}
\begin{remark}
When $L$ is the unit interval, then $\otimes$ is the usual multiplication operator.
\end{remark}

The cardinality of a set is equal to the number of elements the set contains. 
Clearly, the previous definition is not in spirit with this assumption. 
However, hybrid sets may contain elements that occur a negative number of 
times. Thus, one may think that we should take this fact under consideration 
when computing the cardinality of a hybrid set and, more generally, the 
cardinality of an $L$-fuzzy hybrid set. So, it makes sense to
introduce the notion of a {\em strong} cardinality defined as follows:   
\begin{definition}
Assume that $\mathscr{A}$ is an $L$-fuzzy hybrid set that draws elements from 
a universe $X$. Then its {\emph strong} cardinality is defined as follows:
\begin{displaymath}
\card\mathscr{A} = \sum_{x\in X}\mathscr{A}_{\mu}(x)\otimes
|\mathscr{A}_{m}(x)|,
\end{displaymath}
where $|\mathscr{A}_{m}(x)|$ denotes the absolute value of 
$\mathscr{A}_{m}(x)$. 
\end{definition}

For reasons of completeness I give below the definition of the cardinality of 
$L$-multi-fuzzy sets:
\begin{definition}
Assume that $\mathscr{A}$ is an $L$-multi-fuzzy set that draws elements from a 
universe $X$. Then its cardinality is defined as follows:
\begin{displaymath}
\card\mathscr{A} = \sum_{x\in X}\mathscr{A}_{\mu}(x)\otimes\mathscr{A}_{m}(x),
\end{displaymath}
where $\ell\otimes n$ is some binary operator that maps $\ell\in L$ and 
$n\in\mathbb{N}$ to some positive real number (since $n\ge0$).
\end{definition}

In order to complete the presentation of the basic properties of fuzzy hybrid sets, it is
necessary to define the notion of subsethood. Before, going on with this definition, I
will introduce the (new) partial order $\ll$ over $\mathbb{Z}$. In particular, if
$n,m\in\mathbb{Z}$, then 
\begin{eqnarray*}
n\ll m &\equiv& (n=0) \vee\\
       &      & \Bigl((n>0) \wedge (m>0) \wedge (n\le m)\Bigr) \vee\\
       &      & \Bigl((n<0) \wedge (m>0)\Bigr) \vee\\
       &      & (|n|\le|m|).
\end{eqnarray*}
Note that here $\wedge$ and $\vee$ denote the classical 
logical conjunction and disjunction operators, respectively. In addition, the symbols $\le$ and $<$ 
are the well-known ordering operators, and $|n|$ is the absolute value of $n$. 
\begin{example}
From the previous definition it should be obvious that $0\ll n$, for all $n\in\mathbb{Z}$. Also,
$3\ll 4$, $-3\ll 4$, and $-4\ll -3$.
\end{example}

But what kind of structure is
the pair $(\mathbb{Z},\ll)$? The answer is easy with the help of the following
result: 
\begin{proposition}
The relation $\ll$ is a partial order.
\end{proposition}
\begin{proof}
I have to prove that the relation $\ll$ is reflexive, antisymmetric and transitive:
\begin{description}
\item[Reflexivity] Assume that $a\in\mathbb{Z}$. Then if $a=0$, $a\ll a$ from the
first part of the disjunction. If $a<0$, then $a\ll a$ from the fourth part of the
disjunction and if $a>0$, then $a\ll a$ from the second part of the disjunction.
\item[Antisymmetry] Assume that $a,b\in\mathbb{Z}$, $a\ll b$, and $b\ll a$. Then
if $a=0$ this implies that $b=0$ and so $a=b$. If $a<0$, then it follows that $b<0$,
$|a|\le|b|$, and $|b|\le|a|$, which implies that $a=b$. Similarly, if $a>0$, then it
follows that $b>0$, $a\le b$, and $b\le a$, which implies that $a=b$.
\item[Transitivity] Assume that $a,b,c\in\mathbb{Z}$, $a\ll b$, and $b\ll c$.
Then if $a=0$, then clearly $a\ll c$. If $a<0$ and $b<0$, then either $c<0$ or $c>0$, but
since $|b|\le|c|$, this implies that $a\ll c$. If $a>0$ and $b>0$, then $c>0$ and since 
$b\le c$ this implies that $a\ll c$. If $a<0$ and $b>0$, then since $b\ll c$, this implies
that $c>0$, which means that $a\ll c$.
\end{description}
\end{proof} 
Note that $a\gg b$ is an alternative form of $b\ll a$, which will be used in the rest
of this paper. Let us now proceed with the definition of the notion of subsethood for 
$L$-fuzzy hybrid sets:
\begin{definition}
Assume that $\mathscr{A},\mathscr{B}:X\rightarrow L\times\mathbb{Z}$ are two $L$-fuzzy 
hybrid sets. Then $\mathscr{A}\subseteq\mathscr{B}$ if and only if 
$\mathscr{A}_{\mu}(x)\sqsubseteq\mathscr{B}_{\mu}(x)$ and
$\mathscr{A}_{m}(x)\ll\mathscr{B}_{m}(x)$ for all $x\in X$.
\end{definition} 
Remark that for all $\ell_1,\ell_2\in L$, $\ell_1\sqsubseteq\ell_2$ if $\ell_1$ is ``less 
than or equal'' to $\ell_2$ in the sense of the partial order defined over $L$. The
definition of subsethood for $L$-multi-fuzzy sets is more straightforward:
\begin{definition}
Assume that $\mathscr{A},\mathscr{B}:X\rightarrow L\times\mathbb{N}$ are two 
$L$-multi-fuzzy sets. Then $\mathscr{A}\subseteq\mathscr{B}$ if and only if 
$\mathscr{A}_{\mu}(x)\sqsubseteq\mathscr{B}_{\mu}(x)$ and
$\mathscr{A}_{m}(x)\le\mathscr{B}_{m}(x)$ for all $x\in X$.
\end{definition} 

\section{Basic Set Operations}
The basic operations between sets are their union and their intersection. A 
third operation, viz. set sum, is meaningful only for multisets. Also, since 
both $L$-multi-fuzzy sets and $L$-fuzzy hybrid sets are actually 
generalizations of fuzzy sets, one should be able to define the $\alpha$-cuts 
of such sets. I will start by defining the basic set operations between 
$L$-multi-fuzzy sets.

\subsection{Set Operations of $L$-Multi-Fuzzy Sets}
Let me first present the definitions of union and intersection of 
$L$-multi-fuzzy sets:
\begin{definition} Assuming that $\mathscr{A},\mathscr{B}:X\rightarrow L\times
\mathbb{N}$ are two $L$-multi-fuzzy  sets, then their union, denoted 
$\mathscr{A}\cup\mathscr{B}$, is defined as follows:
\begin{displaymath}
\bigl(\mathscr{A}\cup\mathscr{B}\bigr)(x)=\biggl(\mathscr{A}_{\mu}(x)
                                                 \sqcup\mathscr{B}_{\mu}(x),
                                                 \max\Bigl\{\mathscr{A}_{m}(x),
                                                             \mathscr{B}_{m}(x)\Bigr\}
                                           \biggr),
\end{displaymath}
where $a\sqcup b$ is the join of $a,b\in L$.
\end{definition}
\begin{definition} Assuming that $\mathscr{A},\mathscr{B}:X\rightarrow L\times
\mathbb{N}$ are two $L$-multi-fuzzy sets, then their intersection, denoted 
$\mathscr{A}\cap\mathscr{B}$, is defined as follows:
\begin{displaymath}
\bigl(\mathscr{A}\cap\mathscr{B}\bigr)(x)=\biggl(\mathscr{A}_{\mu}(x)
                                                 \sqcap\mathscr{B}_{\mu}(x),
                                                 \min\Bigl\{\mathscr{A}_{m}(x),
                                                             \mathscr{B}_{m}(x)\Bigr\}
                                           \biggr),
\end{displaymath}
where $a\sqcap b$ is the meet of $a,b\in L$.
\end{definition}
I will now define the sum of two $L$-multi-fuzzy sets:
\begin{definition}
Suppose that $\mathscr{A},\mathscr{B}:X\rightarrow L\times\mathbb{N}$ 
are two $L$-multi-fuzzy sets. Then their sum, denoted 
$\mathscr{A}\uplus\mathscr{B}$, is defined as follows:
\begin{displaymath}
\bigl(\mathscr{A}\uplus\mathscr{B}\bigr)(x)=\biggl(\mathscr{A}_{\mu}(x)
                                                 \sqcup\mathscr{B}_{\mu}(x),
                                                 \mathscr{A}_{m}(x) + \mathscr{B}_{m}(x)
                                           \biggr).
\end{displaymath}
\end{definition}
Although it is crystal clear, it is necessary to say that $\sqcup$ and $\sqcap$ are 
operators that are part of the definition of the frame $L$. And as such they have a number
of properties (e.g., they are idempotent, etc., see~\cite[p.~15]{vickers90} for details) 
that, naturally, affect the properties of the operations defined so far. Indeed, these 
operations have the following properties:
\begin{theorem}
For any three $L$-multi-fuzzy sets $\mathscr{A},\mathscr{B},\mathscr{C}:X\rightarrow
L\times\mathbb{N}$ the following equalities hold:
\begin{enumerate}
\item Commutativity:
\begin{eqnarray*}
\mathscr{A}\cup\mathscr{B}&=&\mathscr{B}\cup\mathscr{A}\\
\mathscr{A}\cap\mathscr{B}&=&\mathscr{B}\cap\mathscr{A}\\
\mathscr{A}\uplus\mathscr{B}&=&\mathscr{B}\uplus\mathscr{B};
\end{eqnarray*}

\item Associativity:
\begin{eqnarray*}
\mathscr{A}\cup(\mathscr{B}\cup\mathscr{C})&=&
(\mathscr{A}\cup\mathscr{B})\cup\mathscr{C}\\
\mathscr{A}\cap(\mathscr{B}\cap\mathscr{C})&=&
(\mathscr{A}\cap\mathscr{B})\cap\mathscr{C}\\
\mathscr{A}\uplus(\mathscr{B}\uplus\mathscr{C})&=&
(\mathscr{A}\uplus\mathscr{B})\uplus\mathscr{C};
\end{eqnarray*}

\item Idempotency:
\begin{eqnarray*}
\mathscr{A}\cup\mathscr{A}&=&\mathscr{A}\\
\mathscr{A}\cap\mathscr{A}&=&\mathscr{A};
\end{eqnarray*}

\item Distributivity:
\begin{eqnarray*}
\mathscr{A}\cap(\mathscr{B}\cup\mathscr{C})&=&
(\mathscr{A}\cap\mathscr{B})\cup(\mathscr{A}\cap\mathcal{C})\\
\mathscr{A}\cup(\mathscr{B}\cap\mathscr{C})&=&
(\mathscr{A}\cup\mathscr{B})\cap(\mathscr{A}\cup\mathscr{C});
\end{eqnarray*}

\item Distributivity of sum:
\begin{eqnarray*}
\mathscr{A}\uplus(\mathscr{B}\cup\mathscr{C})&=&
(\mathscr{A}\uplus\mathscr{B})\cup(\mathscr{A}\uplus\mathscr{C})\\
\mathscr{A}\uplus(\mathscr{B}\cap\mathscr{C})&=&
(\mathscr{A}\uplus\mathscr{B})\cap(\mathscr{A}\uplus\mathscr{C});
\end{eqnarray*}
\end{enumerate}
\end{theorem}
\begin{proof}
\begin{enumerate}
\item Although this is easy, I will prove all cases:
\begin{eqnarray*}
\bigl(\mathscr{A}\cup\mathscr{B}\bigr)(z)&=&\biggl(
  \mathscr{A}_{\mu}(z)\sqcup\mathscr{B}_{\mu}(z),
  \max\Bigl\{\mathscr{A}_{m}(z),
             \mathscr{B}_{m}(z)\Bigr\}\biggr)\\
&=&
                       \biggl(
  \mathscr{B}_{\mu}(z)\sqcup
             \mathscr{A}_{\mu}(z),
  \max\Bigl\{\mathscr{B}_{m}(z),
             \mathscr{A}_{m}(z)\Bigr\}\biggr)\\
&=& \bigl(\mathscr{B}\cup\mathscr{A}\bigr)(z)\\[5pt]
\bigl(\mathscr{A}\cap\mathscr{B}\bigr)(z)&=&\biggl(
  \mathscr{A}_{\mu}(z)\sqcap\mathscr{B}_{\mu}(z),
  \min\Bigl\{\mathscr{A}_{m}(z),
             \mathscr{B}_{m}(z)\Bigr\}\biggr)\\
&=&
                       \biggl(
  \mathscr{B}_{\mu}(z)\sqcap\mathscr{A}_{\mu}(z),
  \min\Bigl\{\mathscr{B}_{m}(z),
             \mathscr{A}_{m}(z)\Bigr\}\biggr)\\
   &=& \bigl(\mathscr{B}\cap\mathscr{A}\bigr)(z)\\[5pt]
\bigl(\mathscr{A}\uplus\mathscr{B}\bigr)(z)&=&\biggl(
  \mathscr{A}_{\mu}(z)\sqcup\mathscr{B}_{\mu}(z),
  \mathscr{A}_{m}(z)+\mathscr{B}_{m}(z)\biggr)\\
&=&
                       \biggl(
   \mathscr{B}_{\mu}(z)\sqcup\mathscr{A}_{\mu}(z),
   \mathscr{B}_{m}(z)+ \mathscr{A}_{m}(z)\biggr)\\
   &=& \bigl(\mathscr{B}\uplus\mathscr{A}\bigr)(z)
\end{eqnarray*}

\item I will prove only the first case as the others can be proved similarly:
\begin{eqnarray*}
(\mathscr{A}\cup(\mathscr{B}\cup\mathscr{C}))(z) &=&
\biggl(\mathscr{A}_{\mu}(z)\sqcup\Bigl(\mathscr{B}_{\mu}(z)\sqcup\mathscr{C}_{\mu}(z)
                                \Bigr),
       \max\Bigl\{\mathscr{A}_{m}(z),\max\Bigl\{\mathscr{B}_{m}(z),\mathscr{C}_{m}\Bigr\}
           \Bigr\}\biggr)\\
&=& \biggl(\Bigl(\mathscr{A}_{\mu}(z)\sqcup\mathscr{B}_{\mu}(z)\Bigr)\sqcup
           \mathscr{C}_{\mu}(z),
       \max\Bigl\{\max\Bigl\{\mathscr{A}_{m}(z),\mathscr{B}_{m}(z)\Bigr\},
           \mathscr{C}_{m}\Bigr\}
           \biggr)\\
&=& ((\mathscr{A}\cup\mathscr{B})\cup\mathscr{C})(z)
\end{eqnarray*}
\item As in the previous case, I will prove only the first case as the other can be proved
 similarly:
\begin{eqnarray*}
(\mathscr{A}\cup\mathscr{A})(z) &=& \Bigl(\mathscr{A}_{\mu}(z)\sqcup\mathscr{A}_{\mu}(z),
                                    \max\{\mathscr{A}_{m}(z),\mathscr{A}_{m}(z)\}\Bigr)\\
&=& \Bigl(\mathscr{A}_{\mu}(z),\mathscr{A}_{m}(z)\Bigr)\\
&=& \mathscr{A}(z)      
\end{eqnarray*}
\item The proof of this case follows from the fact that the following equalities are
true for the any three elements of a frame:
\begin{eqnarray*}
x\sqcap(y\sqcup z) &=& (x\sqcap y)\sqcup(x\sqcap y)\\
x\sqcup(y\sqcap z) &=& (x\sqcup y)\sqcap(x\sqcup y)
\end{eqnarray*}
\item As with the previous case the proof for this case follows from the fact 
that for any $x,y,z\in\mathbb{N}$ the following equalities hold:
\begin{eqnarray*}
x+\max\{y,z\} &=& \max\{x+y, x+z\}\\
x+\min\{y,z\} &=& \min\{x+y, x+z\}
\end{eqnarray*}
\end{enumerate}
\end{proof}

The $\alpha$-cut of a fuzzy subset is just a crisp set. Similarly, the 
$\alpha$-cut of
an $L$-multi-fuzzy set has to be a multiset. Indeed, if $[x]_{n}$ denotes a
multiset that consists of only $n$ copies of $x$, the following definition 
is in spirit with the general theory of fuzzy sets:
\begin{definition}
Suppose that $\mathscr{A}$ is an 
$L$-multi-fuzzy set with universe the set $X$, and
that $\alpha\in L$. Then the $\alpha$-cut of $\mathscr{A}$, denoted by
${}^{\alpha}\kern-0.25em\mathscr{A}$, is the multiset
\begin{displaymath}
{}^{\alpha}\kern-0.25em\mathscr{A}=
\bigcup_{\substack{x\in X \\ \alpha\sqsubseteq\mathscr{A}_{\mu}(x)}} [x]_{\mathscr{A}_{m}(x)}.
\end{displaymath}
\end{definition}
Not so surprisingly, the properties of the $\alpha$-cut of $L$-multi-fuzzy sets
are similar to those of plain fuzzy sets. These properties are summarized 
below:
\begin{theorem}\label{alpha-proof}
Assume that $\mathscr{A}$ and $\mathscr{B}$ are two $L$-multi-fuzzy sets with universe
the set $X$. Then the following properties hold:
\begin{enumerate}
\item if $\alpha\sqsubseteq\beta$, then ${}^{\alpha}\kern-0.25em\mathscr{A}\supseteq 
{}^{\beta}\kern-0.25em\mathscr{A}$ and
\item $^{\alpha}\kern-0.15em(\mathscr{A}\cap
\mathscr{B})={}^{\alpha}\kern-0.25em
\mathscr{A}\cap{}^{\alpha}\kern-0.05em\mathscr{B}$,
$^{\alpha}\kern-0.15em(\mathscr{A}\cup\mathscr{B})={}^{\alpha}\kern-0.25em
\mathscr{A}\cup{}^{\alpha}\kern-0.05em\mathscr{B}$, and
$^{\alpha}\kern-0.15em(\mathscr{A}\uplus\mathscr{B})=
{}^{\alpha}\kern-0.25em\mathscr{A}\uplus{}^{\alpha}\kern-0.05em\mathscr{B}$.
\end{enumerate}
\end{theorem}
\begin{proof}
\begin{enumerate}
\item Let $x\in X$ and $\alpha\sqsubseteq\beta$. If 
$\mathscr{A}_{\mu}(x)\not\sqsubseteq\beta$, then 
${}^{\alpha}\kern-0.25em\mathscr{A}(x)={}^{\beta}\kern-0.25em\mathscr{A}(x)$. 
If $\alpha\sqsubseteq\mathscr{A}_{\mu}(x)\sqsubseteq\beta$, then
${}^{\alpha}\kern-0.25em\mathscr{A}(x)\ge{}^{\beta}\kern-0.25em\mathscr{A}(x)$.
If $\alpha\not\sqsubseteq\mathscr{A}_{\mu}(x)$, then 
${}^{\alpha}\kern-0.25em\mathscr{A}(x)={}^{\beta}\kern-0.25em\mathscr{A}(x)=
0$. Thus, for all possible cases  ${}^{\alpha}\kern-0.25em\mathscr{A}(x)\ge
{}^{\beta}\kern-0.05em\mathscr{A}(x)$, which means that
${}^{\alpha}\kern-0.25em\mathscr{A}\supseteq{}^{\beta}\kern-0.25em\mathscr{A}$.
\item Assume that ${}^{\alpha}(\mathscr{A}\cap\mathscr{B})(x)=n$. Then
this means that
\begin{displaymath}
\min\Bigl\{\mathscr{A}_{m}(x),\mathscr{B}_{m}(x)\Bigr\}=n.
\end{displaymath}
Also, it implies that 
$(\mathscr{A}\cap\mathscr{B})_{\mu}(x)\sqsupseteq\alpha$ and hence
$\mathscr{A}_{\mu}(x)\sqcap\mathscr{B}_{\mu}(x)\sqsupseteq\alpha$. From this,
one can immediately deduce that $\mathscr{A}_{\mu}(x)\sqsupseteq\alpha$ and 
$\mathscr{B}_{\mu}(x)\sqsupseteq\alpha$. Suppose now that 
$\mathscr{A}_{m}(x)=n_{1}$ and $\mathscr{B}_{m}(x)=n_{2}$. Then this means
that ${}^{\alpha}\kern-0.25em\mathscr{A}(x)=n_{1}$ and
${}^{\alpha}\kern-0.05em\mathscr{B}(x)=n_{2}$ and so
\begin{displaymath}
\min\Bigl\{{}^{\alpha}\kern-0.25em\mathscr{A}(x),
{}^{\alpha}\kern-0.05em\mathscr{B}(x)\Bigr\}=n.
\end{displaymath}
\end{enumerate}
\end{proof}

\subsection{Set Operations of $L$-Fuzzy Hybrid Sets}
Loeb has shown that the set of all subsets of a given hybrid set with the subsethood 
relation do not form a lattice. This means that if
$f$ and $g$ are two hybrid sets, then if they have lower bounds, they do not necessarily
have a greatest lower bound. Similarly, if $f$ and $g$ have upper bounds, then they do not
necessarily have a lowest upper bound. Practically, this means that given two hybrid sets
$f$ and $g$, one cannot define their union and their intersection. Fortunately, the sum
of hybrid sets is a well-defined operation. Thus, we can easily extend this definition
as follows:
\begin{definition}
Assume that $\mathscr{A},\mathscr{B}:X\rightarrow L\times\mathbb{Z}$ 
are two $L$-fuzzy hybrid sets. Then their sum, denoted 
$\mathscr{A}\uplus\mathscr{B}$, is defined as follows:
\begin{displaymath}
\bigl(\mathscr{A}\uplus\mathscr{B}\bigr)(x)=\biggl(\mathscr{A}_{\mu}(x)
                                                 \sqcup\mathscr{B}_{\mu}(x),
                                                 \mathscr{A}_{m}(x) + \mathscr{B}_{m}(x)
                                           \biggr).
\end{displaymath}
\end{definition}

Let $\{f_{i}\}$ denote a finite collection of hybrid sets with a common universe
$X$, where each of these sets contains repeated occurrence of only one element 
$x_{i}\in X$. In addition, let us insist that no two $f_{i}$ and $f_{j}$ will have common
elements. Also, let us denote with $\uplus_{i}f_{i}$ the unique hybrid set that is the sum
of all $f_{i}$. With these preliminary definitions, the road for the following definition 
has been paved:
\begin{definition}
Suppose that $\mathscr{A}$ is an 
$L$-fuzzy hybrid set with universe the set $X$ and that 
$\alpha\in L$. Then the $\alpha$-cut of $\mathscr{A}$, denoted by
${}^{\alpha}\kern-0.25em\mathscr{A}$, is the hybrid set $\uplus_{i}f_{i}$, where
$f_{i}(x_{i})=\mathscr{A}_{m}(x)$ iff $\alpha\sqsubseteq\mathscr{A}_{\mu}(x)$, for all
$x_{i}\in X$.
\end{definition}
The $\alpha$-cut of $L$-fuzzy hybrid sets has the following properties:
\begin{theorem}
Assume that $\mathscr{A}$ and $\mathscr{B}$ are two $L$-fuzzy hybrid sets 
with universe the set $X$. Then the following properties hold:
\begin{enumerate}
\item if $\alpha\sqsubseteq\beta$, then 
${}^{\alpha}\kern-0.25em\mathscr{A}\supseteq {}^{\beta}\kern-0.25em\mathscr{A}$
\item  $^{\alpha}\kern-0.15em(\mathscr{A}\uplus\mathscr{B})=
{}^{\alpha}\kern-0.25em\mathscr{A}\uplus{}^{\alpha}\kern-0.05em\mathscr{B}$.
\end{enumerate}
\end{theorem}
\begin{proof}
The proof is similar to the proof of theorem~\ref{alpha-proof} and is
omitted.
\end{proof}
\section{General Fuzzy P Systems}
In~\cite{syropoulos06} the author has proposed fuzzified versions of P~systems.
The basic idea behind this particular attempt to fuzzify P~systems is the 
substitution of one or all ingredients of a P~system with their fuzzy 
counterparts. From a purely computational point of view, it turns out that only P~systems 
that process multi-fuzzy sets are interesting. The reason being the fact that these systems 
are capable of computing (positive) real numbers. By replacing the multi-fuzzy sets employed 
in the author's previous work with $L$-multi-fuzzy sets, the computational power of the resulting 
P~systems will not be any ``greater,''  nevertheless, these systems may be quite useful in modeling 
living organisms. But, things may get really interesting if we consider P~systems with
$L$-fuzzy hybrid sets, in general. Let us begin with the definition of these systems:
\begin{definition}\label{f:data}
A general fuzzy P system is a construction
\begin{displaymath}
\Pi_{\mathrm{FD}}=(O, \mu, w^{(1)},\ldots, w^{(m)},
     R_1,\ldots, R_m, i_0)
\end{displaymath}
where:
\begin{enumerate}
\item $O$ is an alphabet (i.e., a set of distinct entities) whose elements
      are called \textit{objects};
\item $\mu$ is the membrane structure of degree $m\ge 1$; membranes
      are injectivelly labeled with succeeding natural numbers starting with
      one;
\item $w^{(i)}:O\rightarrow L\times\mathbb{Z}$, $1\le i\le m$,
      are $L$-fuzzy hybrid sets over $O$ that are associated with each region $i$;
\item $R_i$, $1\le i\le m$, are finite sets of multiset rewriting rules
      (called
      \textit{evolution rules}) over $O$. An evolution rule is of the form
      $u\rightarrow v$, $u\in O^{\ast}$ and $v\in O^{\ast}_{\mathrm{TAR}}$,
      where $O_{\mathrm{TAR}} = O\times\mathrm{TAR}$,
      $\mathrm{TAR}=\{\mathrm{here},\mathrm{out}\}\cup\{\mathrm{in}_{j} |
      1\le j\le m\}$. The effect of each rule is the removal of the
      elements of the left-hand side of each rule from the
      ``current'' compartment and the introduction of the elements of
      right-hand side to the designated compartments;
\item $i_0\in\{1,2,\ldots,m\}$ is the label of an elementary membrane (i.e.,
      a membrane that does not contain any other membrane), called the
      \textit{output} membrane.
\end{enumerate}
\end{definition}

The really interesting thing with the systems described in~\cite{syropoulos06} is that I 
haven't managed to find any limits on what can be actually computed. Remember, that
a real number $x\in\mathbb{R}$ is called {\em computable} if there is a computable
sequence $(r_{n})_{n\in\mathbb{N}}$ of rational numbers which converges to $x$ effectively,
that is for all $n\in\mathbb{N}$, $|x-r_{n}|<2^{-n}$ (see~\cite{weihrauch00,zheng01}
for details). In other words, this means that not all real numbers are {\em computable}.
However, one should not forget that the definition of {\em computability} is hard-wired to
the computational capabilities of the Universal Turing Machine and the so called
Church-Turing thesis, which dictates what can be and what cannot be computed. 
Now, the crucial question is whether there are any limits that prohibit the 
computation of certain numbers with fuzzy P~systems? It seems that these system
go beyond the Church-Turing barrier because their set of input values is drastically
larger than that of the Turing machine. However, it is an open problem the determination
of the exact computational power of these systems.
\section{Conclusions}
In this paper I have introduced $L$-multi-fuzzy sets and $L$-fuzzy hybrid sets as well as
their basic operations. In addition, general fuzzy P~systems have been introduced, which 
can be used to compute real numbers. I do not believe that this is something really new---it 
is just another indication that the current theory of computation is simply inadequate to describe 
all computational phenomena. After all, this has been elegantly demonstrated by Stein in his 
thought provoking paper~\cite{stein99}. In addition, I believe that we need a paradigm shift in 
computer science so to encompass new ``phenomena'' and practices. 

\end{document}